\documentclass[a4paper, 12pt]{article}

\usepackage{amsmath}
\numberwithin{equation}{section}
\usepackage{amsthm}
\usepackage{amssymb}
\usepackage{empheq}
\usepackage{physics}
\usepackage{mathtools}
\usepackage{ulem}
\usepackage{color}
\usepackage{dsfont}
\usepackage{tikz}
\usepackage{tikz-cd}
\usepackage{hyperref}
\usepackage[title]{appendix}
\usepackage{floatrow}
\usepackage{caption}
\def\hh{\mathsf{h}}
\def\tr{{\rm tr}}
\def\TT{\mathcal{T}}
\def\MR{\mathcal{M}_r(\mathcal{T})}

\newtheorem{theorem}{Theorem}[section]
\newtheorem{corollary}[theorem]{Corollary}
\newtheorem{proposition}[theorem]{Proposition}
\newtheorem{lemma}[theorem]{Lemma}
\theoremstyle{definition}
\newtheorem{definition}[theorem]{Definition}
\newtheorem{example}[theorem]{Example}

\newcommand\unit{\hbox{\rm 1\kern-2.8truept l}}

\begin{document}

\title{Irreducibility of Quantum Markov Semigroups, uniqueness of invariant states and related properties}
\author{Franco Fagnola, Federico Girotti
\footnote{Dipartimento di Matematica, Piazza Leonardo da Vinci 32, I-20133 Milano, Italy}}
%
%

\maketitle

\abstract{We present different characterizations of the notion of irreducibility for Quantum Markov Semigroups (QMSs) and investigate its relationship with other relevant features of the dynamics, such as primitivity, positivity improvement and relaxation; in particular, we show that irreducibility, primitivity and relaxation towards a faithful invariant density are equivalent when the semigroup admits an invariant density. Moreover, in the case of uniformly continuous QMSs, we present several useful ways of checking irreducibility in terms of the operators appearing in the generator in GKLS form. Our exposition is as much self-contained as possible, we present some well known results with elementary proofs (collecting all the relevant literature) and we derive new ones. We study both finite and infinite dimensional evolutions and we remark that many results only require the QMS to be made of Schwarz maps.}

\section{Introduction}
\label{sec:intro}

Quantum Markov semigroups (QMSs) on $\mathcal{B}(\mathsf{h})$, or dynamical semigroups in the Physics terminology,
are a key structure for describing the dynamics of open quantum systems subject to noise because of the interaction
with the surrounding environment.

The structure and their properties are well known in the case where the Hilbert space $\mathsf{h}$ is finite dimensional
(see e.g. \cite{AlLe,AmFaKo,AmFaKoII, FaRe-Notes,Wolf-QC&O} and references therein).
However, in the case where $\hh$ is infinite dimensional, even if the generator is a bounded operator, certain fundamental
properties are not entirely clear or their proof not infrequently depends on deep or complex results of functional analysis.
This is the case of the relationships between irreducibility, uniqueness and faithfulness of invariant densities (primitivity)
and convergence to invariant densities.

In this paper we present elementary proofs (i.e., requiring as a prerequisite only a basic course in functional analysis)
of some results on irreducibility, its relationship with other properties and their consequences for the dynamics.

As in the classical theory of Markov processes, irreducibility implies uniqueness and faithfulness of invariant densities (if they exist),
and, practically, it is the most direct and effective method to show this. In most cases it is the only one. Moreover, for QMSs
with an invariant density, irreducibility is indeed \emph{equivalent} to its uniqueness and faithfulness (Theorem \ref{th:equiv-irred}).

A consequence of L{\'e}vy's result (see Proposition 1.3 in \cite{An91}) concerning the support of transition probabilities in classical Markov chains is that irreducibility of the semigroup implies that the probability of finding the process in any state at any positive time is strictly positive; the counterpart for QMSs of such a result would be that irreducibility implies that, no matter what is the initial density, the density at any positive time is faithful. This property of the semigroup is often called positivity improvement. The notions of irreducibility and positivity improvement are known to be equivalent for finite dimensional systems (Theorem \ref{th:irr-iff-prim-conv-inv-st}), while the latter is strictly stronger than irreducibility in infinite dimensions, as shown for instance by Example \ref{ex:irred-no-pos-impr}. The converse implication can be proved under some extra assumptions on the semigroup (see Theorem \ref{th:Lell-sa}).

Another feature of quantum dynamics which is relevant for applications is to determine whether the system relaxes towards an invariant density. In finite dimensional systems irreducibility is equivalent to relaxation towards the same faithful invariant density, regardless of the initial density (Theorem \ref{th:irr-iff-prim-conv-inv-st}) and this keeps being true in infinite dimensional systems as well, once one requires the existence of at least an invariant density (otherwise it does not make sense to talk about relaxation in first place), as shown by Theorem \ref{thm:relaxing}; to the best of our knowledge, this result in such generality is new in the literature. In figures 1, 2 and 3 below we make use of diagrams in order to represent implications between the notions of irreducibility, primitivity, positivity improvement (pos. impr.) and relaxation towards a faithful density (relax.+f.d.) under different hypotheses.

\begin{figure}[h]
\tikzset{column sep=small, ampersand replacement=\&}	
\begin{floatrow}
	\centering
\ffigbox[.4\textwidth]{\begin{tikzcd}
\text{Irreducibility} \arrow[r,shift left=0.9ex, Rightarrow]  \arrow[d,shift left=0.9ex,Rightarrow]
\& \text{Primitivity} \arrow[l, shift left=0.9ex,Rightarrow] \arrow[d,shift left=0.9ex,Rightarrow] \\
\text{Pos. impr.}\arrow[r,shift left=0.9ex,Rightarrow]\arrow[u,shift left=0.9ex,Rightarrow]
\&  \text{Relax. + f.d.} \arrow[l,shift left=0.9ex,Rightarrow]\arrow[u,shift left=0.9ex,Rightarrow]
\end{tikzcd}}{\caption{Diagram of known implications when ${\rm dim}(\hh)<+\infty$.}}
\ffigbox[.4\textwidth]{\begin{tikzcd}
\text{Irreducibility} \arrow[r,shift left=0.9ex, Rightarrow]
\& \text{Primitivity} \arrow[l, shift left=0.9ex,Rightarrow] \arrow[d,shift left=0.9ex,Rightarrow] \\
\text{Pos. impr.}\arrow[r,Rightarrow]\arrow[u,Rightarrow]
\&  \text{Relax. + f.d.} \arrow[u,shift left=0.9ex,Rightarrow]
\end{tikzcd}}{\caption{Diagram of known implications when there exists an invariant density. Missing arrows are not known to hold, in general.}}
\end{floatrow}
\end{figure}

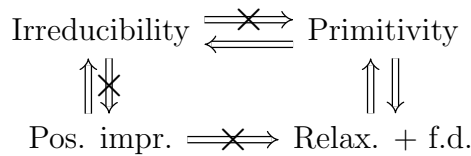
\begin{figure}[h]
\tikzset{column sep=small, ampersand replacement=\&}
\begin{floatrow}
	\centering
\ffigbox[0.8\textwidth]{\begin{tikzcd}
\text{Irreducibility} \arrow[r,shift left=0.9ex, Rightarrow,"\times"{anchor=center,sloped,scale=2}]  \arrow[d,shift left=0.6ex,Rightarrow,"\times"{anchor=center,sloped,scale=2}]
\& \text{Primitivity} \arrow[l, shift left=0.9ex,Rightarrow] \arrow[d,shift left=0.9ex,Rightarrow] \\
\text{Pos. impr.}\arrow[r,Rightarrow,"\times"{anchor=center,sloped,scale=2}]\arrow[u,shift left=0.9ex,Rightarrow]
\&  \text{Relax. + f.d.} \arrow[u,shift left=0.9ex,Rightarrow]
\end{tikzcd}}{
\caption{Diagram of known implications without any assumptions. Crossed out arrows do not hold.}}
\end{floatrow}
\end{figure}

For reasons of space, this work does not treat several other features of QMSs which are related to irreducibility, for instance transience and recurrence (see \cite{FFRR03,CG21,VU}).

Regarding the relevance of our investigation, we remark, for instance, that irreducibility is usually assumed in the study of functional inequalities (see \cite{Ro19} and references therein); moreover, it plays a fundamental role in dissipative state preparation and in the study of the stochastic processes describing the observed measurement outcome and the corresponding conditional state in indirectly continuously monitored quantum systems (\cite{BFPP21,CaYa15,GGG23,MaKu03}). Finally, let us recall that the study of irreducible QMSs helps in understanding non-irreducible ones as well, since they can often be decomposed into irreducible sub-semigroups (see \cite{BaNa,CG21,CaYa,DFSU,Gi22,MA25,VU}).

The structure of the paper is the following one: section \ref{sec:spi} is devoted to introduce the notation and recall the notion of subharmonic projection and irreducibility; a characterization of irreducibility is given in terms of the operators appearing in the infinitesimal generator for uniformly continuous QMSs. In section \ref{sec:iid} we provide several characterizations of irreducibility when the semigroup admits an invariant density (which is a non-trivial assumption only in infinite dimensional systems). Section \ref{sec:findim} establishes the equivalence between irreducibility, positivity improvement and relaxation towards a faithful invariant density for finite dimensional systems. In section \ref{sec:infdim} we dig into the relationship between these three concepts in infinite dimensions.

\section{Subharmonic projections and irreducibility} \label{sec:spi}

In this section we recall the basic definitions and notation.

Let $\hh$ be a complex separable Hilbert space and let $\mathcal{B}(\hh)$ denote the von Neumann algebra
of all bounded linear operators on $\hh$. It is well known that $\mathcal{B}(\hh)$ is isomorphic to the dual of the Banach space
$\mathfrak{I}_1(\hh)$ of trace class operators on $\hh$. We denote by $\Vert\cdot\Vert_1$ the trace norm
on $\mathfrak{I}_1(\hh)$ and by $\Vert\cdot\Vert_\infty$ the operator norm on $\mathcal{B}(\hh)$. Moreover,
we denote by $\operatorname{tr}(\cdot)$ the trace on $\mathfrak{I}_1(\hh)$.

A sequence $(\eta_n)_{n\geq 0}$ converges weakly to $\eta$ in $\mathfrak{I}_1(\hh)$ if the sequence of complex numbers
$(\tr(\eta_n\, x))_{n\geq 0}$ converges to $\tr(\eta\, x)$ for all
$x\in \mathcal{B}(\hh)$. A sequence $(x_n)_{n\geq 0}$ converges weakly$^*$ to $x$ in $\mathcal{B}(\hh)$ if the sequence of complex numbers
$(\tr(\eta\, x_n))_{n\geq 0}$ converges to $\tr(\eta\, x)$ for all
$\eta\in \mathfrak{I}_1(\hh)$.

An element of $\mathfrak{I}_1(\hh)$ is called a \emph{density} if it is positive and
has unit trace. By abuse of terminology, a density is often called state and identified with the
normal state $x\mapsto\tr(\rho x)$ on $\mathcal{B}(\hh)$. The \emph{support}
of a density $\rho$ is the closure of the range of $\rho$. The \emph{support projection} is the orthogonal projection
in $\hh$ onto the range of $\rho$. A density is called \emph{faithful} if its support is the whole
space $\hh$, i.e. its support projection is equal to the identity operator $\unit$.

A \emph{quantum Markov semigroup} (QMS) $\TT=(\TT_t)_{t\geq 0}$ on $\mathcal{B}(\hh)$ is
a weakly$^*$-continuous semigroup of completely positive, identity preserving
normal (i.e. continuous with respect to the weak$^*$ topology on $\mathcal{B}(\hh)$) maps on $\mathcal{B}(\hh)$.
A QMS is the adjoint semigroup (\cite{EnNa} 5.14 p. 63) of a $c_0$ semigroup
$\TT_*=(\TT_{*t})_{t\geq 0}$ of completely positive, trace preserving maps on $\mathfrak{I}_1(\hh)$.
Let us recall that the $c_0$ semigroup property means that for all trace class operator $\eta$, the map
$t\mapsto \TT_{*t}(\eta)$ is continuous for the trace norm topology,
in an equivalent way for the weak topology on $\mathfrak{I}_1(\hh)$ (Theorem 5.8, p. 40 \cite{EnNa}.

We recall that complete positivity of the maps $(\TT_t)_{t \geq 0}$ implies that they satisfy the Schwarz property (see Corollary 2.8 in \cite{Ch74}), that is
\begin{equation} \label{eq:schwarz}
\TT_t(x^*)\TT_t(x) \leq \TT_t(x^*x), \quad x \in {\cal B}(\hh), \, t \geq 0.\end{equation}

If we consider the class of QMSs which are uniformly continuous, i.e. $t \mapsto \TT_t$ is continuous in operator norm, then the celebrated Gorini-Kossakowski-Surdashan-Lindblad (GKSL) theorem (see \cite{GoKoSu, Li76} or Theorem 30.16 in \cite{Pa}) provides a detailed description of its infinitesimal generator.

\begin{theorem}
If $\TT$ is uniformly continuous, then its infinitesimal generator is of the form
\begin{equation} \label{eq:lindblad}
\mathcal{L}(x)=G^*x+xG+\sum_{\ell \in I}L_\ell^*xL_\ell, \quad x \in {\cal B}(\hh)
\end{equation}
for some bounded linear operators $G,L_\ell \in {\cal B}(\hh)$, $\ell\in I$ with $I$ finite or countable
such that $\sum_{\ell \in I}L_\ell^*L_\ell$ converges in the strong operator topology and
\begin{equation} \label{eq:linop}
G^*+G+\sum_{\ell \in I}L_\ell^*L_\ell=0.
\end{equation}
\end{theorem}
In order to highlight the difference with unitary evolutions one can denote by $-H$ the imaginary part of the operator $G$:
from  \eqref{eq:linop} it is immediate to see that
$$
G=-\mathrm{i}H-\frac{1}{2}\sum_{\ell \in I} L_\ell^* L_\ell
$$
and that the generator can be written in the following form:
\begin{equation}\label{eq:GKSL-H}
\mathcal{L}(x)=\mathrm{i}[H,x]+\frac{1}{2}\sum_{\ell \in I}(L_\ell^*L_\ell x-2L_\ell^*xL_\ell+xL_\ell^*L_\ell),
\end{equation}
where one can easily see that the Hamiltonian dynamics is perturbed by the effect of the interaction with the environment, resulting in the additional term with $L_\ell$s. We recall that the representation of the generator given in equation \eqref{eq:GKSL-H} is not unique (Theorem 30.16 in \cite{Pa}) and we remark that all the statements made in this work that involve the operators appearing in equation \eqref{eq:GKSL-H} hold true for every representation.

A density $\rho$ is called \emph{invariant} for a QMS $\TT$ if $\TT_{*t}(\rho)=\rho$ for all $t\geq 0$ or,
in an equivalent way, $\tr(\rho \TT_t(x))=\tr(\rho x)$ for all $x\in\mathcal{B}(\hh)$.

The support of an invariant density has the following useful property.

\begin{proposition}\label{prop:supp-proj}
The support projection $p$ of an invariant density $\rho$ satisfies $\TT_t(p)\geq p$ for all $t\geq 0$.
\end{proposition}

The proof we report is along the lines of the one in Theorem 1 in \cite{VU}, for instance; a different proof for the finite and infinite dimensional case can be found in \cite{BaNa} (Proposition 9) and in \cite{CaYa} (Proposition 5.1), respectively.

\begin{proof}
Let $p$ be the support projection of an invariant density $\rho$. Clearly $p\rho=\rho p =\rho$. Let us fix $t>0$.
Notice that $p\TT_t(p)p\leq p\TT_t(\unit)p = p$ and $\tr(\rho (p-p\TT_t(p)p))
=\tr(\rho (p-\TT_t(p)))=0$; therefore, since $\rho$ is faithful on the subalgebra $p\mathcal{B}(\hh)p$,
$p\TT_t(p)p=p$.

It follows that the projection $p^\perp=\unit-p$ satisfies $p\TT_t(p^\perp)p=0$ and so positivity of $\TT_t(p^\perp)$
implies $p\TT_t(p^\perp) =\TT_t(p^\perp)p=0$, i.e. $\TT_t(p^\perp) \leq p^\perp$ (Lemma \ref{lem:pos}.
Finally, $\TT_t(p)=\unit-\TT_t(p^\perp) \geq p$.
\end{proof}

We call \emph{subharmonic} any orthogonal projection $p$ that satisfies $\TT_t(p) \geq p$.
If the semigroup is uniformly continuous we have a characterization of subharmonic projections in terms of the operators appearing in the generator.
\begin{proposition} \label{prop:spunif}
Let us consider an orthogonal projection $p$; if $\TT$ is uniformly continuous, the following statements are equivalent:
\begin{enumerate}
\item [1.] $p$ is subharmonic,
\item [2.] $\mathcal{L}(p) \geq 0$,
\item [3.] the support of $p$ is an invariant subspace for $G$ and $L_\ell$'s.
\end{enumerate}
\end{proposition}

We recall that a generalization of the equivalence between $1.$ and $3.$ was proved in the case of weak*-continuous semigroups admitting an infinitesimal generator in a generalized Lindblad form (Theorem III.1 in \cite{FFRR02}). The proof of the equivalence between $3.$ and the alternative characterization of $1.$ corresponding to item $3.$ in Proposition \ref{prop:reductions} can be found in \cite{TiVi08} (see Corollary 1) and \cite{ZhBa24} ( Theorem 5.1).
\begin{proof}
Notice that we can rephrase the third item as
$$p^\perp G p=0, \quad p^\perp L_\ell p=0 \quad \forall \ell \in I.$$

\noindent $1. \Rightarrow 2.$ It follows from differentiating $\TT_t(p) \geq p$ at $t=0$.

\noindent $2. \Rightarrow 1.$ Since $\TT_t$ is positive, $2.$ implies that $\TT_t(\mathcal{L}(p))\geq 0$ and the statement follows integrating.

\noindent $2. \Rightarrow 3.$ Notice that $\mathcal{L}(p) \geq 0$ is equivalent to $\mathcal{L}(p^\perp) \leq 0$; therefore, one has
$$0 \geq p\mathcal{L}(p^\perp)p=\sum_{\ell \in I} (p^\perp L_\ell p)^*(p^\perp L_\ell p) \geq 0.$$
Hence, one obtains that $p^\perp L_\ell p=0$ for all $\ell \in I$ and $p\mathcal{L}(p^\perp)p=0$. The negativity of $\mathcal{L}(p^\perp)$ implies that $\mathcal{L}(p^\perp)=p^\perp\mathcal{L}(p^\perp)p^\perp$ (Lemma \ref{lem:pos}) and from $p\mathcal{L}(p^\perp)=0$ one has that $p^\perp Gp=0$.

\noindent $3. \Rightarrow 2.$ Notice that
\begin{eqnarray*}
\mathcal{L}(p^\perp)&=& G^*p^\perp+p^\perp G+\sum_{\ell \in I}L_\ell^*p^\perp L_\ell \\
 & = & p^\perp G^*p^\perp+p^\perp G p^\perp+\sum_{\ell \in I}p^\perp L_\ell^*p^\perp L_\ell p^\perp\\
&= & p^\perp \mathcal{L}(p^\perp) p^\perp \leq 0,
\end{eqnarray*}
where the last inequality follows from the fact that
$p^\perp \TT_t(p^\perp) p^\perp \leq p^\perp \TT_t(\unit)p^\perp\leq   p^\perp$.
\end{proof}

\begin{definition}\label{def:irred}
A completely positive map $\Phi$ is \emph{irreducible} if it does not admit non-trivial subharmonic projections, i.e. if $p$ is a projection and $\Phi(p)\geq p$, then $p=0$ or $p=\unit$.

A QMS $\TT$ is \emph{irreducible} if $(\TT_t)_{t \geq 0}$ do not admit non-trivial common subharmonic projections, i.e.
if $p$ is a projection and $\TT_t(p)\geq p$ for all $t\geq 0$, then $p=0$ or $p=\unit$.
\end{definition}

Notice that irreducibility of a QMS is an a priori weaker requirement than irreducibility of the single maps belonging to the semigroup. The definition we just gave of irreducibility for QMSs is the natural generalization of the same notion for classical Markov semigroups and was introduced in \cite{Da70}. Notice that subharmonic projections allow to {\it reduce the dynamics} in the sense stated by the following result.
\begin{proposition} \label{prop:reductions}
The following are equivalent:
\begin{enumerate}
\item $p$ is a subharmonic projection,
\item $\TT$ preserves the algebra $p^\perp{\cal B}(\hh)p^\perp$, in the sense that
$$\TT_t(x)=p^\perp\TT_t(x)p^\perp, \quad x \in p^\perp{\cal B}(\hh)p^\perp;$$
\item $\TT_*$ preserves $p\mathfrak{I}_1(\hh)p$, in the sense that
$$\TT_{*t}(x)=p\TT_{*t}(x)p, \quad x \in p\mathfrak{I}_1(\hh)p.$$
\end{enumerate}
\end{proposition}
\begin{proof}
First of all, notice that it suffices to prove $2.$ and $3.$ for positive elements.

$1. \Rightarrow 2.$ Let us consider $x \geq 0$ such that $x=p^\perp x p^\perp$, then for all $t \geq 0$
$$0 \leq \TT_t(x) \leq \TT_t(\|x\|p^\perp) \leq \|x\|p^\perp$$
and we can conclude thanks to the positivity of $\TT_t(x)$.

$2. \Rightarrow 1.$ Notice that
$$\TT_t(p^\perp)=p^\perp \TT_t(p^\perp)p^\perp\leq p^\perp \TT_t(\unit)p^\perp=p^\perp.$$

$2. \Rightarrow 3.$ We will prove that $2.$ implies $3.$ and the opposite implication can be shown in the same way. Let us consider $\eta \geq 0$, $\eta=p\eta p$, then $2.$ implies that for every $x \in p^\perp {\cal B}(\hh)p^\perp$ one has that for all $t \geq 0$
$$\tr(\TT_{*t}(\eta)x)=\tr(\eta \TT_t(x))=\tr(\eta p^\perp \TT_t(x) p^\perp)=0,$$
which implies that $\TT_{*t}(\eta)=p\TT_{*t}(\eta) p$ and we are done.

\end{proof}

Thanks to Proposition \ref{prop:spunif}, we can rephrase the notion of irreducibility in terms of the operators appearing
in the infinitesimal generator. Before stating the Theorem, let us introduce the following notation: given a set of
bounded linear operators $S \subseteq \mathcal{B}(\hh)$ let us denote by $\mathcal{A}(S)$ the smallest algebra closed
in the strong operator topology which contains $S$, i.e. the closure in the strong operator topology of the following set:
$$
{\rm span}_{\mathbb{C}}\left\{\,A_1 A_2 \cdots A_{n-1} A_n \,\mid\, \, A_1, \dots, A_n \in S,\,n \geq 1\,\right\}.$$
\begin{proposition} \label{prop:ucirr}
If $\TT$ is uniformly continuous, the following are equivalent:
\begin{enumerate}
\item [1.] $\TT$ is irreducible,
\item [2.] $G$ and $L_\ell$'s have no common invariant subspaces other than $\{0\}$ and $\hh$,
\item [3.] For every nonzero $v \in \hh$,
$\left\{Av\,\mid\,\, A \in \mathcal{A}(\{G, L_\ell\}_{\ell \in I})\right\}=\hh$.
\end{enumerate}
If $\hh$ is finite dimensional, the previous statements are also equivalent to
\begin{enumerate}
\item [4.] $\mathcal{A}(\{G, L_\ell\}_{\ell \in I})=\mathcal{B}(\hh)$.
\end{enumerate}
\end{proposition}

For references regarding the equivalence between $1.$ and $4.$, we refer to \cite{ZhBa24} (Theorem 5.2), where they also provide some examples showing that condition $4.$ (and, consequently, irreducibility) is not equivalent to $\mathcal{A}(\{H, L_\ell\}_{\ell \in I})=\mathcal{B}(\hh)$, nor $\mathcal{A}(\{H, L_\ell,L_\ell^*\}_{\ell \in I})=\mathcal{B}(\hh)$, as it has sometimes been mistakenly assumed in the literature.

We remark that in operator theoretical terminology, one would restate item $3.$ saying that every nonzero vector $v \in \hh$ is cyclic for ${\cal A}(\{G,L_\ell\}_{\ell \in I})$. It is easy to realize that $3.$ and $4.$ do not change if one considers ${\cal A}(\{G,L_\ell, \unit\}_{\ell \in I})$ in place of ${\cal A}(\{G,L_\ell\}_{\ell \in I})$.

\begin{proof}
The equivalence between $1.$ and $2.$ is an immediate consequence of Proposition \ref{prop:spunif}, while the one between $2.$ and $3.$ follows from the fact that every nonzero invariant subspace for $\{G, L_\ell\}_{\ell \in I}$ contains a subspace of the form
$\{Av: \, A \in \mathcal{A}(\{G, L_\ell\}_{\ell \in I})\}$ for some nonzero $v \in \hh$, which, in turn, is an invariant subspace for $G$ and $L_\ell$s itself.

Finally, the equivalence between $2.$ and $4.$ when $\hh$ is finite dimensional is Burnside theorem
(see \cite{LoRo} for a simple short proof).
\end{proof}

We remark that an extension of Theorem 2 in \cite{Sp} to the infinite dimensional setting is a simple corollary of Proposition \ref{prop:ucirr}. It provides a sufficient condition for irreducibility. Before stating the result, we need to introduce some more notation. We say that the set ${\rm span}_{\mathbb{C}}\{L_\ell: \ell \in I\}$ is self adjoint if for all $\ell\in I$, there exist complex numbers
$(z_{\ell\,j})_{j\in I}$ such that $L_\ell^*=\sum_{j\in I}z_{\ell j}L_j$, the series being strongly convergent; given a set $S \subseteq {\cal B}(\hh)$, we denote by $S^\prime$ the commutant of $S$, i.e. the set $\{x \in {\cal B}(\hh):\, [x,y]=0, \, \forall y \in S\}$.

\begin{corollary}
Let us assume that $\TT$ is uniformly continuous and that the set ${\rm span}_{\mathbb{C}}\{L_\ell: \ell \in I\}$ is self adjoint; if ${\rm span}_{\mathbb{C}}\{L_\ell: \ell \in I\}^\prime=\mathbb{C}\unit$, then $\TT$ is irreducible.
\end{corollary}
\begin{proof}
Let us assume that $G$ and $L_\ell$s have a common invariant subspace and let $p$ denote the corresponding orthogonal projection. One has that $L_\ell p=p L_\ell p$ for every $\ell \in I$. Since $L^*_\ell$ can be seen a (possibly infinite) linear combination of $L_j$s, one has that $L^*_\ell p=pL_\ell^* p$ as well and $[L_\ell,p]=0$. Therefore, we just showed that $p \in  {\rm span}_{\mathbb{C}}\{L_\ell: \ell \in I\}^\prime=\mathbb{C}\unit$ and, consequently, $p \in \{0,\unit\}$.
\end{proof}

We just mention that another sufficient condition can be found in finite dimensional systems in terms of the so called Kossakowski matrix (see item 3. in Corollary 7.2 in \cite{Wolf-QC&O}).

\section{Irreducibility and invariant densities} \label{sec:iid}

In this section we focus on QMSs admitting an invariant density and we will mainly investigate the relationship between irreducibility and primitivity.
\begin{definition}
A QMS $\TT$ is said to be \emph{primitive} if it admits a unique faithful invariant density.
\end{definition}
This convention is usually adopted in the literature that studies functional inequalities for QMSs in finite dimensional systems (see for instance \cite{GaoRou, KaTe}), while, when dealing with discrete time dynamics, it usually refers to, in addition, assuming the absence of eigenvalues of modulus $1$ other than $1$ itself. In fact, the two notions are equivalent for QMSs in finite dimensions (see Theorem \ref{th:irr-iff-prim-conv-inv-st}).

An old result due to A. Frigerio (\cite{Fr77} Theorem 1) shows that, if a QMS $\TT$ has a faithful invariant density,
this is unique (hence the semigroup is primitive) if and only if $\TT$ is irreducible. The proof in \cite{Fr77} is not elementary because it depends on
ergodic theorems for $c_0$ semigroups and the GNS construction which are not covered in standard functional analysis courses.
In this section we present an elementary proof. We start with the following lemma which is notoriously very useful in the analysis
of the asymptotic behavior of semigroups.

\begin{lemma}\label{lem:invariant}
Any weak$^*$ limit of a sequence of Ces\`aro means $t_m^{-1}\int_0^{t_m} \TT_r(x)\mathrm{d}r$ as
$t_m$ tends to infinity is a fixed point for all maps $\TT_t$.
\end{lemma}

\begin{proof} Denote by $y$ the weak$^*$ limit of the sequence. For all $m,s\geq 0$ we have
\begin{eqnarray*}
  \TT_s \left(\frac{1}{t_m}\int_0^{t_m} \TT_r(x)\mathrm{d}r\right)
   &=&  \frac{1}{t_m} \int_s^{t_m+s} \TT_r(x)\mathrm{d}r \\
   &=&   \frac{1}{t_m} \int_0^{t_m} \TT_r(x)\mathrm{d}r \\
   &+& \frac{1}{t_m} \int_{t_m}^{t_m+s} \TT_r(x)\mathrm{d}r
  -\frac{1}{t_m} \int_{0}^{s} \TT_r(x)\mathrm{d}r
\end{eqnarray*}
The integrals of the last two terms are bounded because
\begin{eqnarray*}
\left\Vert \int_{t_m}^{t_m+s} \TT_r(x)\mathrm{d}r\right\Vert
& \leq & \int_{t_m}^{t_m+s} \Vert\TT_r(x)\Vert\mathrm{d}r \leq s \Vert x \Vert \\
\left\Vert \int_{0}^{s} \TT_r(x)\mathrm{d}r\right\Vert
& \leq & \int_{0}^{s} \Vert\TT_r(x)\Vert\mathrm{d}r \leq s \Vert x \Vert
\end{eqnarray*}
Therefore, taking weak$^*$ limits as $t_m$ tends to infinity, we find $\TT_s(y)=y$.
\end{proof}

Before stating the next result, let us introduce the set of fixed points of $\TT$:

$${\cal F}(\TT):=\{x \in {\cal B}(\hh):\, \text{for all }t \geq 0, \,\TT_t(x)=x \}.$$

\begin{theorem}\label{th:equiv-irred}
Let $\TT$ be a QMS with a faithful invariant density $\rho$.

\bigskip An orthogonal projection $p$ is subharmonic if and only if $p \in {\cal F}({\cal T})$.

\bigskip \noindent Moreover, the following are equivalent:
\begin{itemize}
  \item[{1.}] $\TT$ is irreducible,
  \item[{2.}] The set of fixed points is trivial, i.e. $\mathcal{F}(\TT)=\mathbb{C}\unit$,
  \item[{3.}] $\rho$ is the unique invariant density, i.e. $\TT$ is primitive.
\end{itemize}
If $\TT$ is uniformly continuous, the previous items are equivalent to
\begin{itemize}
  \item[{4.}] $\{G,L_\ell,L_\ell^*\}^{\prime}=\{H,L_\ell,L_\ell^*\}^{\prime}=\mathbb{C}\unit$.
\end{itemize}
\end{theorem}
Notice that, since subharmonic projections are in fact harmonic when there exists a faithful invariant density, in this particular case the notion of irreducibility that we are considering (usually referred to as Davies irreducibility, see \cite{Da70}) coincides with the one introduced by Evans (\cite{Ev77}).

\begin{proof}

If an orthogonal projection $p$ is in ${\cal F}({\cal T})$, then it is trivially a subharmonic projection as well. Let us show that the converse is also true. Let $p$ be a subharmonic projection. By the invariance of $\rho$ we have
$\tr(\rho\TT_t(p)) = \tr(\rho p)$, therefore
$\tr\left(\rho(\TT_t(p)- p)\right)=0$.
Since $\rho$ is faithful and $\TT_t(p)\geq p$, it follows that $\TT_t(p)= p$ for all $t\geq 0$,
namely $p$ belongs to $\mathcal{F}(\TT)$.

$1.\Rightarrow 2.$ The set $\mathcal{F}(\TT)$ of fixed points is a von Neumann subalgebra of $\mathcal{B}(\hh)$.
Indeed, it is a $^*$-subalgebra of $\mathcal{B}(\hh)$ because, if $x,y\in\mathcal{F}(\TT)$, then
\[
\TT_t((x+\theta y)^*(x+\theta y)) \geq \TT_t((x+\theta y)^*)\TT_t((x+\theta y))=(x+\theta y)^*(x+\theta y)
\]
for all $\theta\in\mathbb{C}$, $t\geq 0$. Moreover, by the invariance of $\rho$,
\[
\tr\left(\rho\TT_t((x+\theta y)^*(x+\theta y))-(x+\theta y)^*(x+\theta y)\right)=0
\]
and so, by faithfulness of $\rho$, $\TT_t((x+\theta y)^*(x+\theta y))=(x+\theta y)^*(x+\theta y)$.
By the arbitrariness of $\theta$, $\TT_t(x^*y)=x^*y$ for all $t\geq 0$.
Furthermore, $\mathcal{F}(\TT)$ is clearly closed for the weak$^*$ topology and so it is a von Neumann algebra which is
generated by its projections.
Any projection $p$ in $\mathcal{F}(\TT)$ is also a subharmonic projection, and so it is trivial by irreducibility.
Therefore $\mathcal{F}(\TT)=\mathbb{C}\unit$.

$2.\Rightarrow 1.$ Let $p$ be a subharmonic projection. We proved that, under the assumption that there exists a faithful normal invariant state, $p$ belongs to $\mathcal{F}(\TT)=\mathbb{C}\unit$. This shows that $p$ is either $0$ or $\unit$,
and $\TT$ is irreducible.

$2.\Rightarrow 3.$  Note that, for all $x\in \mathcal{B}(\hh)$,
the operators $t^{-1}\int_0^t \TT_s(x)\mathrm{d}s$ ($t>0$) are contained in the ball of radius $\Vert x\Vert_\infty$
which is sequentially weakly$^*$ compact, therefore we can find a divergent sequence $(t_m)_{m\geq 0}$ weakly$^*$ converging
to some limit. This belongs to $\mathcal{F}(\TT)$ by Lemma \ref{lem:invariant}, and so it is a multiple $c(x)\unit$ of the identity operator.

Note that, for all $t>0$,
\[
\tr(\rho x)= \tr\left(\rho\, \frac{1}{t}\int_0^t \TT_{r}(x)\mathrm{d}r\right)
\]
and so
\[
c(x) = \lim_{m\to\infty}\tr\left(\rho\, \frac{1}{t_m}\int_0^{t_m} \TT_{r}(x)\mathrm{d}r\right)
=\tr(\rho x).
\]
It follows that, for any invariant density $\eta$, we have
\[
\tr\left(\eta x\right)
= \tr\left(\eta\,
\left(\operatorname{w}^*-\lim_{m\to +\infty}\frac{1}{t_m}\!\int_0^{t_m} \TT_{r}(x)\mathrm{d}r\right) \right)
=c(x)\tr(\eta)=\tr(\rho x)
\]
and $\eta=\rho$ by the arbitrariness of $x$.

$3.\Rightarrow 1.$ By contradiction, if $\TT$ is not irreducible, then there exists a non-trivial subharmonic projection $p$. We showed that, since there exists a faithful normal invariant state, $p \in {\cal F}({\cal T})$.
Moreover, in order to prove $1 \Rightarrow 2.$, we demostrated that for every $x \in {\cal F}({\cal T})$ and for every $t \geq0$, one has that
$${\cal T}_t(x^*x)={\cal T}_t(x^*){\cal T}_t(x), \quad {\cal T}_t(xx^*)={\cal T}_t(x){\cal T}_t(x^*).$$
Lemma \ref{lem:Schw-map} implies that for every $t \geq 0$ and $x \in {\cal  B}(\hh)$, the following holds true:
$${\cal T}_t(qxr)=q{\cal T}_t(x)r, \quad q,r \in \{p,p^\perp\}.$$
By duality, the same can be shown to hold for the predual semigroup, i.e. for every $t \geq 0$ and $\omega \in \mathfrak{I}_1(\hh)$
$${\cal T}_{*t}(q\omega r)=q{\cal T}_{*t}(\omega )r, \quad q,r \in \{p,p^\perp\}.$$
Therefore, $p\rho p/\tr(\rho p)$ is another, non-faithful, invariant state
for $\TT$ contradicting $3.$

\end{proof}

Leveraging on Theorem \ref{th:equiv-irred}, one can actually show that the equivalence between $1.$ and $3.$ keeps holding true even without assuming the faithfulness of $\rho$.

\begin{proposition} \label{prop:irrprim}
Let $\TT$ be a QMS admitting an invariant density $\rho$. The following are equivalent:
\begin{enumerate}
\item[{1.}] $\TT$ is irreducible,
  \item[{2.}] $\rho$ is faithful and it is the unique invariant density, i.e. $\TT$ is primitive.
  \item[{3.}] $\TT_{t_0}$ is irreducible for all $t_0>0$.
  \item[{4.}]  $\TT_{t_0}$ is irreducible for some $t_0>0$.
\end{enumerate}
\end{proposition}
The content of Proposition \ref{prop:irrprim} is shown in \cite{Wolf-QC&O} (Proposition 7.5) for finite dimensional systems. The proofs of some of the implications that we present below are simpler and hold in infinite dimensions as well.
\begin{proof}
$2. \Rightarrow 1.$ If ${\cal T}$ admits a unique faithful invariant state, then ${\cal T}$ is irreducible thanks to the implication $3. \Rightarrow 1.$ in Theorem \ref{th:equiv-irred}.

$1. \Rightarrow 2.$ On the other hand, let us now suppose that ${\cal T}$ is irreducible. The support projection of $\rho$ is subharmonic by Proposition \ref{prop:supp-proj} and, by irreducibility, it must coincide with the identity operator, therefore $\rho$ is faithful. The statement follows now from the implication $1. \Rightarrow 3.$ of Theorem \ref{th:equiv-irred}.

Before proving the equivalence of $1.$ and $3.$, we remark that the proof of the equivalence between $1.$ and $2.$ in Theorem \ref{th:equiv-irred} and the equivalence between $1.$ and $2.$ that we just proved hold for discrete time QMSs as well, i.e. powers of a single normal completely positive and unital map. Moreover, one can check that the proof of Theorem \ref{thm:relaxing} is independent from the equivalence between $1.$ and $3.$ that we are about to prove, therefore it does not cause any inconsistency to use it in this proof.

$1. \Rightarrow 3.$ If $\TT$ is irreducible, then we just showed that the invariant density $\rho$ is unique and strictly positive. Theorem \ref{thm:relaxing} implies that for every intial density $\eta$ one has
$$\lim_{t \rightarrow +\infty}\TT_{*t}(\eta)=\rho$$
in trace class norm. By contradiction, let us suppose that there exists a time $t_0>0$ such that $\TT_{t_0}$ is not irreducible, then, by Theorem \ref{th:equiv-irred} (applied to discrete time QMSs), $\TT_{t_0}$ admits another invariant density $\rho^\prime$ other than $\rho$, however we get to the contradiction that
$$\rho^\prime=\lim_{n \rightarrow +\infty}\TT_{*nt_0}(\rho^\prime)=\rho.$$

$3. \Rightarrow 4.$ is trivial.

$4. \Rightarrow 1.$ If $\TT_{t_0}$ is irreducible, then Theorem \ref{th:equiv-irred} (applied to discrete time QMSs) ensures that $\rho>0$ and that
$${\cal F}((\TT_{nt_0})_{n \geq 0}):=\{x \in {\cal B}(\hh): \TT_{nt_0}(x)=x, \, n \geq 0\}=\mathbb{C}\unit;$$
since ${\cal F}(\TT)\subseteq {\cal F}((\TT_{nt_0})_{n \geq 0}),$ we can conclude.
\end{proof}

\medskip
It is well known that, when the Hilbert space $\hh$ has finite dimension, one can always find an invariant density.
In fact, given an initial density $\eta$, thanks to the trace-norm compactness of the set of densities, one can find a
divergent sequence $(t_m)_{m\geq 1}$ such that $\Big(t_m^{-1}\int_0^{t_m}\TT_{*s}(\eta)\mathrm{d}s\Big)_{m\geq 1}$
converges in trace norm to a density $\rho$. It is then verified, as in the proof of Lemma \ref{lem:invariant}, that this limit is
an invariant density. We can then remove the hypothesis of the existence of an invariant density from Proposition \ref{prop:irrprim}
and find the following Corollary.

\begin{corollary}\label{cor:h-fin-dim}
If the Hilbert space $\hh$ is finite dimensional, then the following statements are equivalent:
\begin{enumerate}
\item[{1.}] $\TT$ is irreducible,
  \item[{2.}] $\TT$ is primitive.
  \item[{3.}] $\TT_{t_0}$ is irreducible for all $t_0>0$.
  \item[{4.}]  $\TT_{t_0}$ is irreducible for some $t_0>0$.
\end{enumerate}
\end{corollary}

\section{Positivity improvement and relaxation (finite dimension)} \label{sec:findim}

When $\hh$ is finite dimensional, irreducibility is  equivalent to other properties that are frequently encountered in applications (see for instance \cite{Sp}).

\begin{definition}
A QMS is called
\begin{itemize}
  \item \emph{Positivity improving} if $x\geq 0$ and $x\not=0$ implies $\TT_t(x) > 0 $
  for all $t>0$.
  \item \emph{Ergodic (relaxing)} if for every density $\eta$, there exists a desnity $\rho$ such that
  $$\lim_{t \rightarrow +\infty}\TT_{*t}(\eta)=\rho$$
in trace class norm.
\end{itemize}
\end{definition}

We remark that $\rho$ appearing in the definition of ergodicity is an invariant density, since for every $t \geq 0$
$${\cal T}_{*t}(\rho)={\cal T}_{*t} \left (\lim_{s \rightarrow +\infty}{\cal T}_{*s}(\rho)\right )= \lim_{s \rightarrow +\infty}{\cal T}_{*(t+s)}(\rho)=\rho.$$

It is worth noticing that positivity improvement can be defined in an equivalent way for the predual QMS $\TT_*$.

One can prove the following result (\cite{Wolf-QC&O} Proposition 7.5 p.125), which clarifies all the implications between the various notions introduced so far in the finite dimensional case.

\begin{theorem}\label{th:irr-iff-prim-conv-inv-st}
If the Hilbert space $\hh$ is finite dimensional, then the following are equivalent
\begin{enumerate}
  \item $\TT$ is irreducible,
  \item $\TT$ is primitive,
  \item $\TT_{t_0}$ is irreducible for all $t_0>0$.
  \item  $\TT_{t_0}$ is irreducible for some $t_0>0$.
  \item $\TT$ is primitive and $\sigma(\mathcal{L})\cap {\rm i}\mathbb{R}=\{0\}$.
  \item $\TT$ is ergodic with limit state $\rho$ being faithful.
  \item $\TT$ is positivity improving.
\end{enumerate}
\end{theorem}

\begin{proof}
The first four items are equivalent by Corollary \ref{cor:h-fin-dim}.

$3. \Rightarrow 5.$ If the map $\TT_t$ for all $t >0$, then, by the Perron-Frobenius theorem \ref{th:frobenius}, the eigenspace of an eigenvalue $\lambda_t$
with $|\lambda_t|=1$ is one-dimensional and eigenvectors are multiple of a unitary operator $u$.
Applying $\TT_s$ to the identity $\TT_t(u)=\lambda_t u$, we find
\[
\TT_{t}\left(\TT_s(u)\right) = \TT_{s}\left(\TT_t(u)\right)
= \lambda_t   \TT_s(u)
\]
which implies that $\TT_s(u)=\theta_s u$ for some $\theta_s$ with $|\theta_s|\leq 1$ for all $s>0$.
Continuity in $s$, the semigroup property and $\theta_0=1,\, \theta_t=\lambda_t$ with $|\lambda_t|=1$,
imply $\theta_s=\mathrm{e}^{\mathrm{i}\theta s}$ for some $\theta\in\mathbb{R}$. But $\mathrm{e}^{\mathrm{i}\theta s}$
belongs to a finite subgroup of $\mathbb{T}$ for all $s$ if and only if $\theta=0$.
Since all the other eigenvalues of $\TT_t$ have modulus strictly smaller than $1$ for all $t\geq 0$, it
follows that the generator $\mathcal{L}$ has eigenvalue $0$ with eigenvector $\unit$ and all the other
eigenvalues with strictly negative real part.

$5. \Rightarrow 6.$ The pre-dual generator $\mathcal{L}_*$ has the same
eigenvalues as $\mathcal{L}$, therefore $\TT_{*t}=\mathrm{e}^{t\mathcal{L}_*}$ has all eigenvalues of the
form $\mathrm{e}^{\lambda t}$ with $\operatorname{Re}(\lambda)<0$ except $1$. Therefore
$\TT_{*t}(\eta)$ converges to $\rho$ as $t$ goes to $+\infty$ for all initial density $\eta$.

$6. \Rightarrow 7.$ If we prove the statement for pure states, we are done.
Fix a normalized vector $v \in \hh$; since in finite dimensions the set of strictly positive matrices is open and,
by assumption, $\TT_t(|v\rangle\langle v|)$ converges to $\rho>0$ for large times, there exists $t_1>0$ such that $\TT_t(|v\rangle\langle v|)>0$ as well for every $t \geq t_1$. Lemma \ref{lem:supp} shows that the support of $\TT_t(|v\rangle\langle v|)$ for $t>0$ is given by $P_t \mathcal{S}(v)$, for a certain subspace
$\mathcal{S}(v)$ of $\hh$. Since $P_t$ is bijective, the dimension of the
support of $\TT_t(|\psi\rangle\langle v|)$ for $t>0$ is constant and it is maximum for $t \geq t_1$, we are done.

7. $\Rightarrow$ 1. If $\TT$ is positivity improving then, for all non-trivial projection $p$,
the kernel of $\TT_t(p^\perp)$ is $\{0\}$ and so the inequality $\TT_t(p^\perp)\leq p^\perp$
cannot hold.
\end{proof}


\section{Infinite dimensional case} \label{sec:infdim}

In this section we analyze the relationship between irreducibility and all the notions introduces so far -primitivity, positivity improvement and ergodicity- for a QMS on $\mathcal{B}(\hh)$
with $\hh$ infinite dimensional.
In this situation two fundamental difficulties emerge: $\TT$ may not possess any invariant density and $\TT$ could fail to be uniformly continuous.

\subsection{Primitivity}

There are many irreducible QMSs without invariant states, as the following simple example shows; therefore we cannot hope of removing the assumption in Proposition \ref{prop:irrprim} as in the case of finite dimensional systems.
\begin{example}\label{ex:bd-chain}
Let $\hh=\ell^2(\mathbb{N};\mathbb{C})$ be the space of square summable complex valued sequences which is naturally
isomorphic to the Fock space $\Gamma(\mathbb{C})$, with canonical
orthonormal basis $(e_k)_{k\geq 0}$. Let $S$ be the right shift defined by $Se_k=e_{k+1}$ and
let $L_1=S$, $L_2=S^*$. Consider the QMS generated by
\[
\mathcal{L}(x) = -\frac{1}{2}\sum_{\ell=1}^2 \left(L_\ell^*L_\ell x -2L_\ell^* x L_\ell + x L_\ell^* L_\ell\right)
\]
This is positivity improving because, if we denote $G=-(L_1^*L_1+L_2^*L_2)/2$, Lemma \ref{lem:supp} ensures that the support of any
$\TT_{*t}(|\psi\rangle\langle\psi|)$ is
\[
\mathrm{e}^{t G}\operatorname{span}_{\mathbb{C}}\left\{\, \psi, \
 \delta_{G}^{m_{1}}(L_{\ell_1})\delta_{G}^{m_{2}}(L_{\ell_2})\cdots
 \delta_{G}^{m_{n}}(L_{\ell_n})\psi \, \right\}
\]
for $n\geq 1$, $\ell_1\dots,\ell_n\in I$, $m_1,\dots, m_n \in \mathbb{N}$, which is $S,S^*$ (and so also $G$) invariant.
The only $S,S^*$ invariant subspaces of $\hh$ are trivial.
Therefore the QMS generated by $\mathcal{L}$ is positivity improving and irreducible. \\
To check that it has no invariant density we note that the restriction of $\mathcal{L}$
to the abelian von Neumann algebra $L^\infty(\mathbb{N};\mathbb{C})$ generated by rank-one projections $|e_n\rangle\langle e_n|$
determines a classical continuous time Markov chain without invariant density.
Let $\mathcal{E}$ be the conditional expectation
\[
\mathcal{E}:\mathcal{B}(\hh) \mapsto L^\infty(\mathbb{N};\mathbb{C}) \qquad
\mathcal{E}(x) = \sum_{n\geq 0} \left\langle e_n,x\, e_n\right\rangle\, |e_n\rangle\langle e_n|\, .
\]
Indeed, since $Ge_n=-e_n$ for $n>0$, $Ge_0=-e_0/2$, we have
\begin{eqnarray*}
(\mathcal{L}\circ\mathcal{E})(x) &=& \langle e_0,xe_0\rangle \left(|e_1\rangle \langle e_1|-|e_0\rangle \langle e_0| \right) \\
  & + & \sum_{n\geq 1}\langle e_n,xe_n\rangle \left(|e_{n+1}\rangle \langle e_{n+1}|-2|e_n\rangle \langle e_n|
  +|e_{n-1}\rangle \langle e_{n-1}|\right)\\
(\mathcal{E}\circ\mathcal{L})(x)  &=& \sum_{n\geq 0} \langle e_n,(G^*x+S^*xS+SxS^*+xG)e_n\rangle |e_n\rangle\langle e_n|\\
 & = & \left(\langle e_1,xe_1\rangle-\langle e_0,xe_0\rangle\right) |e_0\rangle \langle e_0| \\
 & + & \sum_{n\geq 1} \left(\langle e_{n+1},xe_{n+1}\rangle+\langle e_{n-1},xe_{n-1}\rangle
 -2 \langle e_n,xe_n\rangle\right)|e_n\rangle \langle e_n|
\end{eqnarray*}
and $(\mathcal{L}\circ\mathcal{E})(x) -(\mathcal{E}\circ\mathcal{L})(x)=0$. As a consequence
$\TT\circ\mathcal{E}=\mathcal{E}\circ\TT$ for all $t\geq 0$ and the restriction of
$\TT$ to $L^\infty(\mathbb{N};\mathbb{C})$ is the Markov semigroup of a classical
time-continuous birth-and-death chain with generator
\begin{eqnarray*}
(Af)(n) & = & \left(f(n+1)-f(n)\right) +\left(f(n-1)-f(n)\right)\quad \text{for} \quad n>0,\\
(Af)(0) & = & \left(f(1)-f(0)\right).
\end{eqnarray*}
This has no classical invariant density.

Suppose that $\rho$ is an invariant density for $\TT$. Note that, for all $x\in\mathcal{B}(\hh)$
and all $t\geq 0$, we have
\[
\tr(\mathcal{E}_*(\rho)\,x)=
\tr(\rho\,\mathcal{E}(x))=\tr(\rho\,\TT_t(\mathcal{E}(x)))
=\tr(\rho\,\mathcal{E}(\TT_t(x)))=\tr(\mathcal{E}_*(\rho)\TT_t(x))
\]
It follows that $\mathcal{E}_*(\rho)=\sum_{n\geq 1} \langle e_n,\rho e_n\rangle |e_n\rangle \langle e_n|$ would determine
the  classical invariant density $(\langle e_n,\rho e_n\rangle)_{n\geq 1}$ of the above birth-and-death chain.
Therefore $\TT$ has no invariant density.

\end{example}

\subsection{Ergodicity}

Let us now move to the property of the semigroup of relaxing towards a unique faithful density. When $\hh$ is finite dimensional, this follows from the Perron-Frobenius
theorem as shown in the proof Theorem \ref{th:irr-iff-prim-conv-inv-st}. First, we will provide an elementary proof for Ces\`aro means, i.e.
$$\frac{1}{t}\int_0^t \TT_{*s}(\eta)ds.$$

\begin{lemma}\label{lem:dom-dens}
Let $\TT$ be a QMS with a faithful invariant density $\rho$. For all initial density $\eta$ dominated by $\rho$
there  exists a divergent sequence $(t_m)_{m\geq 1}$ such that
$\left(t_m^{-1}\int_0^{t_m}\TT_{*r}(\eta)\mathrm{d}r\right)_{m\geq 1}$
converges in trace norm to a density $\eta_\infty$.
\end{lemma}

\begin{proof}
Let $b>0$ be a constant such that $\eta\leq b\rho$.
Denote $\eta_t=t^{-1}\int_0^{t}\TT_{*r}(\eta)\mathrm{d}r$. Let $(e_j)_{j\geq 0}$ be an orthonormal basis
with respect to which $\rho$ is diagonal.

For all $n>0$ let $p_n$ be a finite dimensional spectral projection
of $\rho$ such that $\tr(p_n^\perp\rho)<1/bn$.

Note that, by finite dimensionality of $p_n$ for all $n>0$, $(p_n\eta_t p_n)_{t\geq 0}$ are non-negative matrices
with trace norm less or equal than $1$. Therefore we can find a sequence $(t_n(m))_{m\geq 0}$ such that
$(p_n\eta_{t_n(m)} p_n)_{m\geq 0}$ is trace norm converging.

By the diagonal method we can find a sequence $(t_n)_{n> 0}$ such that each $(p_n\eta_{t_n}p_n)_{n>0}$
trace norm converges as follows. For $n=1$ let $(t_1(m))_{m\geq 0}$ be a strictly increasing
sequence such that $p_1\eta_{t_1(m)}p_1$ converges to $\eta^{(1)}=p_1\eta^{(1)}p_1$.
Proceeding by induction, for all $n>0$, we can extract from the sequence $(t_n(m))_{m\geq 0}$
such that $p_n\eta_{t_n(m)}p_n$ converges to $\eta^{(n)}=p_n\eta^{(n)}p_n$, another subsequence
$(t_{n+1}(m))_{m\geq 0}$ such that  $p_{n+1}\eta_{t_{n+1}(m)}p_{n+1}$ converges to
$\eta^{(n+1)}=p_{n+1}\eta^{(n+1)}p_{n+1}$. Choosing $t_n:=t_n(n)$ we are done.

Clearly, $p_n\eta^{(n')}p_n = \eta^{(n)}$ for all $n'>n$. Therefore, for all $n'>n$,
\[
\eta^{(n)}-\eta^{(n')} = \left(\eta^{(n)}-\eta^{(n')}\right)p_n^\perp+p_n^\perp\left(\eta^{(n)}-\eta^{(n')}\right)p_n
\]
Notice that for a positive operator $a \geq 0$, one has $\|p^\perp a p\|_1 \leq \|ap^\perp\|_1$ (it can be easily seen by duality), therefore it follows that, for $n'>n$
\[
\left\Vert  \eta^{(n)}-\eta^{(n')}\right\Vert_1 \leq 2\left\Vert  \eta^{(n)}p_n^\perp\right\Vert_1
+2\left\Vert \eta^{(n')}p_n^\perp\right\Vert_1.
\]
Note that $\eta\leq b\rho$ implies $\eta_t\leq b\rho$ for all $t>0$ and, by H\"older inequality and the domination property, we have
\begin{eqnarray*}
 \left\Vert \eta^{(n)}p_n^\perp\right\Vert_1
   & = & \left\Vert (\eta^{(n)})^{1/2}(\eta^{(n)})^{1/2}p_n^\perp\right\Vert_1   \\
   & \leq & \left\Vert (\eta^{(n)})^{1/2}\right\Vert_2 \left\Vert (\eta^{(n)})^{1/2}p_n^\perp\right\Vert_2 \\
   & = & \left(\tr \left(p_n^\perp\eta^{(n)}p_n^\perp\right)\right)^{1/2} \\
   &\leq & \sqrt{b} \left(\tr \left(p_n^\perp\rho p_n^\perp\right)\right)^{1/2} < 1/\sqrt{n}
\end{eqnarray*}
In the same way we find $\left\Vert \eta^{(n')}p_n^\perp\right\Vert_1<1/\sqrt{n}$, and we get the inequality
\[
\left\Vert  \eta^{(n)}-\eta^{(n')}\right\Vert_1 < 4/\sqrt{n}.
\]
This implies that $(\eta^{(n)})_{n\geq 1}$ is a Cauchy sequence in $\mathfrak{I}_1(\hh)$. Denoting $\eta_\infty$ its limit
we have
\begin{eqnarray*}
\left\Vert \frac{1}{t_n}\int_0^{t_n}\TT_{*r}(\eta)\mathrm{d}r - \eta_\infty \right\Vert_1
& \leq & \left\Vert \frac{1}{t_n}\int_0^{t_n}\TT_{*r}(\eta)\mathrm{d}r - \eta^{(n)} \right\Vert_1
+ \left\Vert  \eta^{(n)}-\eta_\infty\right\Vert_1 \\
& \leq & \left\Vert \frac{1}{t_n}\int_0^{t_n}\TT_{*r}(\eta)\mathrm{d}r - \eta^{(n)} \right\Vert_1
+\frac{4}{\sqrt{n}}
\end{eqnarray*}
which becomes arbitrarily small as $n$ tends to infinity.
\end{proof}

By another approximation argument we can remove domination.

\begin{lemma}\label{lem:seq-ergod}
Let $\TT$ be a QMS with a faithful invariant density $\rho$. For all initial density $\eta$
there  exists a divergent sequence $(t_m)_{m\geq 1}$ such that
$\left(t_m^{-1}\int_0^{t_m}\TT_{*r}(\eta)\mathrm{d}r\right)_{m\geq 1}$
converges in trace norm to a density $\eta_\infty$.
\end{lemma}

\begin{proof}
Since $\rho$ is faithful, the set of densities dominated by $\rho$ is trace-norm dense in the set of all densities
(truncate with spectral projections of $\rho$ and approximate).
For all $\epsilon >0$ let $\eta^{(\epsilon)}$ be an initial density dominated by $\rho$ such that
$\Vert\eta-\eta_\epsilon\Vert_1 < \epsilon/2$ and let $(t_m)_{m\geq 1}$ be a strictly increasing sequence such that
$\left(t_m^{-1}\int_0^{t_m}\TT_{*r}(\eta_\epsilon)\mathrm{d}r\right)_{m\geq 1}$ converges in trace norm
to $\eta^{(\epsilon)}_\infty$ as $m$ tends to infinity. Clearly, by Lemma \ref{lem:dom-dens},
\begin{eqnarray*}
& & \limsup_{m,n\to\infty} \left\Vert \frac{1}{t_m}\int_0^{t_m}\TT_{*r}(\eta)\mathrm{d}r -
\frac{1}{t_n}\int_0^{t_n}\TT_{*r}(\eta)\mathrm{d}r\right\Vert_1 \\
& & < \epsilon + \limsup_{m,n\to\infty} \left\Vert \frac{1}{t_m}\int_0^{t_m}\TT_{*r}(\eta_\epsilon)\mathrm{d}r -
\frac{1}{t_n}\int_0^{t_n}\TT_{*r}(\eta_\epsilon)\mathrm{d}r\right\Vert_1 = \epsilon.
\end{eqnarray*}
The conclusion follows by the arbitrariness of $\epsilon$.
\end{proof}

Finally, we show that convergence of a subsequence implies convergence of the family
$\left(t^{-1}\int_0^{t}\TT_{*r}(\eta)\mathrm{d}r\right)_{t>0}$ by a semigroup theoretic argument.

\begin{lemma}\label{lem:seq-to-cont}
Let $\eta$ be a density. If there exists a divergent sequence $(t_n)_{n\geq 1}$ such that $\left(t_n^{-1}\int_0^{t_n}\TT_{*r}(\eta)\mathrm{d}r\right)_{n\geq 1}$ trace-norm
converges to a density $\eta_\infty$, then $\eta_\infty$ is an invariant density for $\TT$ and
\[
\lim_{t\to\infty}\frac{1}{t}\int_0^t \TT_{*r}(\eta)\mathrm{d}r =\eta_\infty
\]
in trace norm.
\end{lemma}
\begin{proof} $\eta_\infty$ is an invariant density for $\TT$ by Lemma \ref{lem:invariant}.
Let $\eta_t = t^{-1}\int_{0}^t\TT_{*s}(\eta_0)\mathrm{d}s $.
For all $\epsilon>0$ let $n_\epsilon>0$ such that
\[
\left\Vert \eta_{t_n}-\eta_\infty\right\Vert_1=
\left\Vert \frac{1}{t_n}\int_0^{t_n} \TT_{*s}(\eta_0)\mathrm{d}s - \eta_\infty \right\Vert_1 < \epsilon
\]
for all $n\geq n_\epsilon$. Now, for all $t>t_{n_\epsilon}$, $n\geq n_\epsilon$ and $0\leq s\leq t_{n_\epsilon}$, we have
\[
\TT_{*s}\left(\frac{1}{t}\int_0^{t} \TT_{*r}(\eta_0)\mathrm{d}r\right)
- \frac{1}{t}\int_0^{t} \TT_{*r}(\eta_0)\mathrm{d}r
= \frac{1}{t}\int_t^{s+t } \TT_{*r}(\eta_0)\mathrm{d}r -
\frac{1}{t}\int_0^{s  } \TT_{*r}(\eta_0)\mathrm{d}r
\]
and so
\[
\left\Vert \TT_{*s}\left(\eta_t\right)
- \eta_t\right\Vert_1
\leq \frac{2 s}{t} \leq \frac{2 t_{n_\epsilon}}{t}.
\]
Fix a $t(\epsilon)>0$, depending on $\epsilon$ and $n_\epsilon$, such that $2 t_{n_\epsilon}/t \, < \epsilon$.
Integrating for $s\in[0,t_{n_\epsilon}]$ and dividing by $t_{n_\epsilon}$, we find
\[
\left\Vert \frac{1}{t_{n_\epsilon}}\int_0^{t_{n_\epsilon}}\TT_{*s}(\eta_t)\mathrm{d}r- \eta_t \right\Vert_1 <\epsilon.
\]
Notice that
$$\frac{1}{t_{n_\epsilon}}\int_0^{t_{n_\epsilon}}\TT_{*s}(\eta_t)\mathrm{d}r=\frac{1}{t}\int_0^{t}\TT_{*s}(\eta_{t_{n_\epsilon}})\mathrm{d}r.$$
It follows that, for all $t>t(\epsilon)$,
\begin{eqnarray*}
\left\Vert \eta_t - \eta_\infty \right\Vert_1
& = & \left\Vert \frac{1}{t}\int_0^{t} \TT_{*r}(\eta_0)\mathrm{d}r
-\frac{1}{t}\int_0^{t} \TT_{*r}(\eta_\infty)\mathrm{d}r\right\Vert_1  \\
& \leq & \left\Vert \frac{1}{t}\int_0^{t} \TT_{*r}(\eta_0-\eta_{t_{n_\epsilon}})\mathrm{d}r\right\Vert_1
+  \left\Vert \frac{1}{t}\int_0^{t} \TT_{*r}(\eta_\infty-\eta_{t_{n_\epsilon}})\mathrm{d}r\right\Vert_1 \\
& \leq & \left\Vert \eta_t-\eta_{t_{n_\epsilon}}\right\Vert_1
+  \left\Vert \eta_\infty-\eta_{t_{n_\epsilon}}\right\Vert_1 < 2\epsilon.
\end{eqnarray*}
This completes the proof.
\end{proof}

The results of Lemmas \ref{lem:dom-dens}, \ref{lem:seq-ergod}, \ref{lem:seq-to-cont} are summarized in the following

\begin{theorem}\label{th:conv-prim}
Let $\TT$ be a primitive QMS. For every initial density $\eta$,
the family of densities $\left(t^{-1}\int_0^t\TT_{*s}(\eta)\mathrm{d}s\right)_{t>0}$ is trace norm convergent to $\rho$.
\end{theorem}

\medskip
Under the assumptions of Theorem \ref{th:conv-prim}, one may wonder whether $\TT_{*t}(\eta)$ converges to $\rho$.
We have no direct elementary proof, so we will rely on some deep results without proving them and refer to the literature. First of all we need to introduce the so called reversible subalgebra $\mathcal{M}_r(\TT)$ as the weak* closure of the following linear space:
\[
{\rm span}_{\mathbb{C}}\left\{\, x\in\mathcal{B}(\hh)\,\mid\,
\text{for all } t\geq 0, \, \TT_t(x)=\lambda_t x \text{ for some } \lambda_t \in \mathbb{C}, \, |\lambda_t|=1\,\right\}.
\]
One can show the following (see Lemma 1 and Theorem 1 in \cite{CJ20} and Propositions 2.1 and 2.2, noticing that everything extends to the continuous time case).
\begin{theorem}\label{thm:dfs}
Let us assume that ${\cal T}$ admits a faithful invariant density $\rho$, then
$\MR$ is a von Neumann algebra on which $\TT$ acts as a semigroup of *-automorphisms. Moreover, if we denote by $B_1$ and $M_1$ the unit ball of ${\cal B}(\hh)$ and $\MR$, respectively, the following equality holds true:
$$M_1=\bigcap_{t \geq 0} \TT_t(B_1).$$
Finally, there exists a unitary operator $U$ acting on $\hh$ such that
$$U:\hh \rightarrow \bigoplus_{\alpha \in A}\hh_\alpha^L \otimes \hh_{\alpha}^{R}$$ for some denumerable index set $A$ and
\begin{equation} \label{eq:stNT}
\MR=U^*\bigoplus_{\alpha \in A}B(\hh^L_\alpha) \otimes \mathbf{1}_{\hh_\alpha^R} U. \end{equation}
\end{theorem}

The reversible subalgebra plays a fundamental role in our investigations due to the following simple observation.

\begin{lemma} \label{lem:acc}
Let us consider a QMS $\TT$ with a faithful invariant density $\rho$. Then, for every $x \in {\cal B}(\hh)$, $\MR$ contains every accumulation point of the net $(\TT_t(x))_{t \geq 0}$.
\end{lemma}
\begin{proof}
Without loss of generality, we can assume that $\|x\| \leq 1$. Notice that for every $t \geq s$, $\TT_t(B_1) \subseteq \TT_s(B_1)$, indeed $\TT_t(B_1) =\TT_s(\TT_{t-s}(B_1))$ and $\TT_{t-s}(B_1) \subseteq B_1$; hence,
$$\TT_t(B_1) \subseteq \bigcap_{s \leq t}\TT_s(B_1).$$
Therefore, if we consider any accumulation point $\overline{x}$ of the net $(\TT_t(x))_{t \geq 0}$, we have that $\overline{x} \in \bigcap_{t \geq 0}\TT_t(B_1)=M_1$.
\end{proof}

The previous Lemma points out that a natural strategy for proving that primitivity implies relaxation towards a unique faithful invariant density is to show that primitivity implies $\MR=\mathbb{C}\unit$ and this is, indeed, what we are going to do.

\begin{theorem} \label{thm:relaxing}
If ${\cal T}$ is primitive with unique invariant density $\rho>0$, then $\MR=\mathbb{C}\unit$ and for every density $\eta$
$$\lim_{t \rightarrow +\infty}{\cal T}_{*t}(\eta)=\rho$$
in trace class norm.
\end{theorem}
Notice that $\MR=\mathbb{C}\unit$ can be seen as a restatement of item $5.$ in Theorem \ref{th:irr-iff-prim-conv-inv-st} (where, in infinite dimensional systems, one only considers the point spectrum).
\begin{proof}

Without loss of generality and in order to simplify notation, we can assume that $U$ in Theorem \ref{thm:dfs} is the identity operator.

For every $x \in {\cal B}(\hh)$, the set $\{{\cal T}_{t}(x)\}_{t \geq 0}$ is sequentially relatively compact in the weak* topology and, therefore, in order to show that it converges, it is enough to prove that all accumulation points of the net coincide. We will now show that irreducibility of ${\cal T}$ implies that $\MR=\mathbb{C}\unit$.

\textbf{Step 1.} First of all, we will show that primitivity implies that $A$ in Theorem \ref{thm:dfs} is a singleton. Let us denote by $Z$ the center of $\MR$, i.e. the set of all the elements in $\MR$ that commute with all the other elements in $\MR$; it is easy to see from equation \eqref{eq:stNT} that the center is the ${\rm W}^*$-algebra generated by $p_\alpha$s, where $p_\alpha: =p_{\hh_\alpha^L \otimes \hh_\alpha^R}$ for $\alpha \in A$. Notice that ${\cal T}$ maps $Z$ into itself. Indeed, for every $z \in \MR^\prime$ and $y \in \MR$
$${\cal T}_t(z)y={\cal T}_t(z){\cal T}_t(x)={\cal T}_t(zx)={\cal T}_t(xz)={\cal T}_t(x){\cal T}_t(z)=y{\cal T}_t(z),$$
where $x$ is the unique element in $\MR$ such that $y={\cal T}_t(x)$ (one can find such $x$ because Theorem \ref{thm:dfs} ensures that $\TT_t$ acts as a *-automorphism on $\MR$); we just proved that $\TT$ maps $\MR^\prime$ into itself and the claim follows from the fact that $Z=\MR \cap \MR^\prime$. Moreover, since ${\cal T}_t$ acts on $Z$ as a $^*$-automorphism, ${\cal T}_t(p_\alpha)=\sum_{\beta \in A}x_{\beta,\alpha}(t)p_\beta$ is a projection, i.e. $x_{\beta,\alpha}(t) \in \{0,1\}$. By the continuity of $t \mapsto x_{\beta,\alpha}(t)$, one has that ${\cal T}_t(p_\alpha) =p_\alpha$ for every time $t \geq 0$, i.e. $p_\alpha \in {\cal F}({\cal T})$. Since $Z$ is generated by $p_\alpha$s, then $Z \subseteq {\cal F}({\cal T})=\mathbb{C}\unit$, where the last equality is due to primitivity. Therefore, $A$ is a singleton, $\hh=\hh^L \otimes \hh^R$ and $\MR=B(\hh^L) \otimes \unit_{\hh^R}$.

\textbf{Step 2.} We will show that $\rho=\rho_1 \otimes \rho_2 $ for some states $\rho_1$, $\rho_2$ on $B(\hh^L)$ and $B(\hh^R)$, respectively. In step $1.$ we showed that $\TT$ maps $\MR^\prime=\unit_{\hh^L} \otimes B(\hh^R)$ into itself and, by the definition of the reversible subalgebra, we know that $\TT$ maps $\MR= B(\hh^L)\otimes \unit_{\hh^R}$ into itself as well; therefore we can define the two QMSs $\TT^1$ and $\TT^2$ acting on $B(\hh^L)$ and $B(\hh^R)$, respectively, as
\begin{align*}\TT_1(x)\otimes \unit_{\hh^R}:=\TT_t(x \otimes \unit_{\hh^R}), \quad x \in B(\hh^L)\\
\unit_{\hh^L} \otimes \TT_2(y):=\TT_t( \unit_{\hh^L}\otimes y), \quad y \in B(\hh^R).
\end{align*}
Moreover, notice that for all $x \in B(\hh^L)$, $y \in B(\hh^R)$ and $t \geq 0$ one has
\[\begin{split}\TT_t(x \otimes y)&=\TT_t((x \otimes \unit_{\hh^R})( \unit_{\hh^L} \otimes y))=\TT_t(x \otimes \unit_{\hh^R})\TT_t(\unit_{\hh^L} \otimes y)\\
&=(\TT^1_t(x) \otimes\unit_{\hh^R})( \unit_{\hh^L} \otimes  \TT^2_t(y))=\TT^1_t(x) \otimes \TT^2_t(y),\end{split}\]
that is $\TT_t=\TT^1_t \otimes \TT^2_t$.

Let us introduce the notation $\rho_1 = \tr_2(\rho)$ and $\rho_2=\tr_1(\rho)$, where $\tr_1$ ($\tr_2$) is the partial trace with respect to $\hh^L$ (resp. $\hh^R$); then, for every $x \in B(\hh^L)$ and $t \geq 0$ one has
\[\begin{split}
\tr(\rho_1  x )&=\tr(\rho x \otimes \unit_{\hh^R})= \tr(\rho \TT_t(x \otimes \unit_{\hh^R}))\\
&=\tr(\rho \TT^1_t(x) \otimes \unit_{\hh^R})= \tr_1(\rho_1 \TT^1_t(x)),
\end{split}\]
which means that $\rho_1 $ is an invariant density for $\TT^1$; in the same way, one gets that $\rho_2$ is an invariant density for $\TT^2$. Therefore, using that $\TT=\TT^1 \otimes \TT^2$, one gets that $\rho_1 \otimes \rho_2$ is an invariant density for $\TT$ and, by primitivity, $\rho=\rho_1 \otimes \rho_2$.

\textbf{Step 3.} For every $t \geq 0$, ${\cal T}^1_t$ is a $^*$-automorphism; by Wigner theorem (Theorem 14.5 in \cite{Pa}) there exists a unitary or antiunitary operator $U_t$ such that ${\cal T}^1_t(x)=U_t^*x U_t$ for every $x \in B(\hh^L)$; if $1<{\rm dim}(\hh^L)<+\infty$, it is easy to see that one can construct more than one invariant state for the semigroup $\TT^1$ and, consequently, for $\TT$ (attaching $\rho_2$). If ${\rm dim}(\hh^L)=+\infty$, let us consider the spectral resolution of $\rho_1=\sum_{i \in I}\lambda_i p_i$ (notice that $p_i$s are finite dimensional), then one has
$$\rho_1={\cal T}^1_{*t}(\rho_1)=\sum_{i \in I} \lambda_i {\cal T}^1_{*t}(p_i)=\sum_{i \in I} \lambda_i U_t (p_i) U_t^*$$
and $U_t (p_i) U_t^*$ are orthogonal projections summing up (in the strong operator topology) to the identity operator. By uniqueness of the spectral decomposition, $p_i=U_t (p_i) U_t^*={\cal T}_{*t}(p_i)$ for every $i \in I$, therefore $\rho_{i}=p_i/{\rm tr}(p_i)\otimes \rho_2$ is a collection of invariant densities with orthogonal supports and this contradicts primitivity of the semigroup. Therefore ${\rm dim}(\hh^L)=1$ and we showed that $\MR=\mathbb{C}\unit$.

Summing up, for every $x \in {\cal B}(\hh)$ we showed that every accumulation point of the net $\{{\cal T}_{t}(x)\}_{t \geq 0}$ is of the form $\alpha \unit$ for some $\alpha \in \mathbb{C}$ (Lemma \ref{lem:acc} ensures that every accumulation point is in $\MR$ and we showed that $\MR$ is trivial); however, notice that $\alpha$ is completely determined by $x$ since for every $t \geq 0$, ${\rm tr}(\rho x)={\rm tr}(\rho {\cal T}_t(x))$, which implies that $\alpha=\rho(x)$. Therefore we showed that
$${\rm w}^*-\lim_{t \rightarrow +\infty}{\cal T}_{t}(x)={\rm tr}(\rho x)\unit, \quad x \in {\cal B}(\hh),$$
which is equivalent to
$${\rm w}-\lim_{t \rightarrow +\infty}{\cal T}_{*t}(\eta)=\rho$$
for every density $\eta$. It is easy to see that the convergence in the previous equation is equivalent to convergence in trace class norm: indeed, for every $\epsilon>0$, there exists a finite dimensional spectral projection of $\rho$, say $p_\epsilon$, such that ${\rm tr}(\rho p_\epsilon^\perp) < \epsilon$. Using the weak convergence, one can find $t_0$ such that for every $t \geq t_0$ one has
$${\rm tr}({\cal T}_{*t}(\eta) p_\epsilon^\perp) <2\epsilon, \quad \|p_\epsilon {\cal T}_{*t}(\eta)p_\epsilon-p_\epsilon\rho\|_1 <\epsilon.$$
Therefore,
$$\|{\cal T}_{*t}(\eta)-\rho\|_1 \leq \|p_\epsilon {\cal T}_{*t}(\eta)p_\epsilon-p_\epsilon\rho\|_1+2\|({\cal T}_{*t}(\eta)- \rho)p_\epsilon^\perp\|_1 \leq 7\epsilon$$
and we conclude by the arbitrariness of $\epsilon$.

\end{proof}

\subsection{Positivity improvement}

If the Hilbert space $\hh$ is infinite dimensional positivity improvement is stronger than irreducibility
(Example \ref{ex:irred-no-pos-impr} below). However, if $\TT$ is uniformly continuous and under some extra assumptions, they become equivalent (see \cite{HaNa} Theorem 4.2). In order to state this result, we need to introduce the decoherence-free subalgebra (\cite{AmFaKo,FaSaUm,SaUm23}):
$${\cal N}(\TT):=\{x \in {\cal B}(\hh):\TT_t(x^*x)=\TT_t(x^*)\TT_t(x), \,\ t \geq 0\};$$
if the semigroup is uniformly continuous, one can easily see that ${\cal N}(\TT)$ is the biggest von Neumann algebra on which $\TT$ acts as a semigroup of *-automorphisms (see Theorem 3 in \cite{FaSaUm}) and, if the semigroup admits a faithful invariant density, that it coincides with $\MR$ (see Theorem 1 in \cite{CJ20}).

\begin{theorem}\label{th:Lell-sa}
Let us consider a uniformly continuous QMS $\TT$ such that the set ${\rm span}_{\mathbb{C}}\left\{\,L_\ell\,\mid\,\ell\in I\,\right\}$ is self-adjoint.

The following are equivalent:
\begin{enumerate}
\item ${\cal N}(\TT)=\mathbb{C}\unit$,
\item $\TT$ is positivity improving.
\end{enumerate}

If, in addition, we assume that $\TT$ has an invariant density, then the following are equivalent:
\begin{enumerate}
\item[3.] $\TT$ is irreducibile,
\item[4.] $\TT$ is positivity improving.
\end{enumerate}
\end{theorem}

\begin{proof}
The equivalence between $3.$ and $4.$ is a corollary of the equivalence of $1.$ and $2.$: irreducibility implies that the invariant density is faithful (Proposition \ref{prop:irrprim}) and, under the assumption of the existence of a faithful invariant density, irreducibility implies $\MR=\mathbb{C}\unit$ (Theorem \ref{thm:dfs}) and ${\cal N}(\TT)$ coincides with $\MR$ (Theorem 1 in \cite{CJ20}).

Let us introduce the notation $\delta_a(b)=[a,b]$ for every $a,b \in {\cal B}(\hh)$. Since the support of any pure initial density
$|\psi\rangle\langle\psi|$ is characterized as in Lemma \ref{lem:supp}, and the commutant of a self-adjoint set of
operators is trivial if and only if every nonzero vector is cyclic for that set (see \cite{BrRo} Proposition 2.3.8, p. 47)
positivity improvement is equivalent to the commutant of the linear span of the set of operators
\begin{equation}\label{eq:set-supp}
\delta^{m_1}_G(L_{\ell_1}) \cdots
\delta^{m_n}_G(L_{\ell_n}), \quad n \in \mathbb{N}^*, \  m_1,.., m_n \in \mathbb{N}, \ \ell_1, .., \ell_n \in I,
\end{equation}
being trivial; indeed, notice that the linear span of the set of operators in equation \eqref{eq:set-supp} is self-adjoint because of our assumption on the linear space generated by noise operators ${\rm span}_{\mathbb{C}}\left\{\,L_\ell\,\mid\,\ell\in I\,\right\}$.

In order to prove the equivalence between $1.$ and $2.$ it suffices to show that the commutant of the linear span of the operators in equation \eqref{eq:set-supp} coincides with the commutant of the set of operators given by
\begin{equation} \label{eq:NT-comm}
\{\delta_{H}^k(L_\ell),\,\delta_{H}^k(L_\ell^*): \, k \geq 0\}.
\end{equation}
Indeed, one has that (see Proposition 2.1, 2.2 and 2.3 in \cite{DFSU})
$${\cal N}(\TT)=\{\delta_{H}^k(L_\ell),\,\delta_{H}^k(L_\ell^*): \, k \geq 0\}^\prime.$$
To this end, first note that both sets \eqref{eq:NT-comm}
and \eqref{eq:set-supp} contain operators $L_\ell,L_k^*$ and polynomials in $L_\ell,L_k^*$. Then, by the identity
\[
\delta_G(L_\ell) = -\frac{1}{2}\sum_k[L_k^*L_k,L_\ell] -\mathrm{i}\delta_H(L_\ell),
\]
both sets contain operators $\delta_G(L_\ell),\delta_G(L_\ell^*)$ and $\delta_H(L_\ell),\delta_H(L_\ell^*)$.
The conclusion can be extended to iterated commutators of any order by an induction argument as in the proof of
Lemma 4.1 of \cite{HaNa} and the proof is complete.
\end{proof}

To the best of our knowledge, it is yet unclear whether the equivalence between irreducibility and positivity improvement holds for a uniformly continuous irreducible
QMS, without any further assumptions.

\smallskip

Below we show an example of an irreducible QMS which is not positivity improving.
It is not uniformly continuous, neither it has any invariant density, but the unbounded generator $\mathcal{L}$ is representable
in a GKSL form as in \eqref{eq:GKSL-H} allowing an unbounded self-adjoint operator $H$.
The QMS can be constructed from $\mathcal{L}$ by the minimal semigroup method or as a bounded perturbation
of the unitary group generated by the unbounded derivation $x\mapsto \mathrm{i}[H,x]$ (see Appendix B).

\begin{example}\label{ex:irred-no-pos-impr}
Let $\hh=L^2(\mathbb{T};\mathbb{C})$ be the Hilbert space of periodic functions on the unit circle $\mathbb{T}$ and consider
the QMS with generator in a (generalized, since $H$ is unbounded) GKSL form with
\[
H = \mathrm{i}\frac{\mathrm{d}}{\mathrm{d}x} \qquad
(L_\ell f)(x) =2^{-|\ell|-1}\mathrm{e}^{2\mathrm{i}\pi \ell x } f(x),  \quad \ell\in \mathbb{Z}^*=\mathbb{Z}-\{\,0\,\}.
\]
Note that $L_\ell$ are multiples of unitary multiplication operators and
\[
\sum_{\ell\in\mathbb{Z}^*}L_\ell^* L_\ell = \left(\sum_{\ell\in\mathbb{Z}^*}2^{-|\ell|-1}\right)\unit=\unit
\]
and so the semigroup generated by
\[
G = -\frac{1}{2}L_\ell^* L_\ell -\mathrm{i}H
= -\frac{1}{2}\unit + \frac{\mathrm{d}}{\mathrm{d}x}
\]
is
\[
(P_t f)(x) = \mathrm{e}^{-t/2}f(x+t).
\]
 Therefore, the closure of linear combinations of functions
$x\mapsto \mathrm{e}^{2\mathrm{i}\pi \theta\, k x}$ for $k\in\mathbb{N}$ is
dense in $L^2(\mathbb{T})$. As a consequence, applying linear combinations of operators $L_\ell$ to a $u\in L^2(\mathbb{T})$
and taking the closure, one can get any function in $L^2(\mathbb{T})$ with the essential support of $u$. Operators
$P_t=\mathrm{e}^{tG}$ acts by translation and the factor $\mathrm{e}^{-t/2}$ is irrelevant for invariant subspace problems.
It follows from Lemma \ref{lem:support-H-sa} that the support of $\TT_{*t}(|\psi\rangle\langle\psi|)$ cannot be full
for a suitable $\psi$ with small support for small $t$.\\
However, $\TT$ is irreducible because a $P_t$ invariant and $L_\ell$ invariant subspace is trivial. Looking at the restriction of the semigroup to the diagonal algebra with respect to the Fourier basis, one can easily check that $\TT$ does not have any invariant density following a similar reasoning as in Example \ref{ex:bd-chain}.
\end{example}

\section{Conclusion and Outlook}

In this work we presented the state of the art regarding the understanding of irreducibility in open quantum Markovian dynamics and its relationship with several relevant features of the evolution: the existence and uniqueness of invariant densities, positivity improvement and relaxation towards a faithful density; more precisely we showed that, except from positivity improvement, they are all equivalent when the semigroup admits an invariant density. Many proofs have been simplified with respect to the literature and some new results have been added. In finite dimensional systems, everything is well understood and irreducible QMSs share all the relevant properties with their classical counterpart. On the other hand, to the best of our knowledge, the theory in infinite dimensions still has some unclear aspects: while the asymptotic behaviour of Ces\`aro means is well understood (see Theorem  2.3.23 in \cite{GiPhD}), little is known about the net $\{\TT_{*t}(\eta)\}_{t \geq 0}$ for QMSs without an invariant density; another natural question is whether irreducibility implies positivity improvement in the case where the semigroup is uniformly continuous or admits an invariant density.

\section*{Acknowledgement}
The authors are members of GNAMPA-INdAM.
The support of the MUR grant ``Dipartimento di Eccellenza 2023--2027'' of Dipartimento di Matematica,
Politecnico di Milano and ``Centro Nazionale di ricerca in HPC, Big Data and Quantum Computing'' is gratefully  acknowledged.

\appendix

\section{Lemma on positive operators}

\begin{lemma} \label{lem:pos}
Let us consider $x\geq 0$ and $p$ an orthogonal projection; if $pxp=0$, then $x=p^\perp x p^\perp$.
\end{lemma}
\begin{proof}
Let us consider any complex number $z \in \mathbb{C}$ and vectors $u, v \in \hh$ such that $pu=0$ and $pv=v$. The positivity of $x$ implies that
$$0 \leq \langle u+zv, x(u+zv) \rangle= \langle u,xu\rangle +2 \operatorname{Re}(z\langle u, x v \rangle, $$
however, for this to be true for every $z \in \mathbb{C}$ one gets that $\langle u, x v \rangle=0$.
\end{proof}

\section{Peripheral spectrum of irreducible Schwarz unital maps}

The content of this section is based on \cite{FP09,Gr81}. Let us consider $\varPhi$ which is assumed to be a unital Schwarz map (see equation \eqref{eq:schwarz}) on
${\cal B}(\hh)$ and let $\mathcal{D}:
{\cal B}(\hh)\times{\cal B}(\hh)\rightarrow\mathbb{C}$ be the map
\[
\mathcal{D}(a,b) = \varPhi(a^*b)-\varPhi(a^*)\varPhi(b)
\]
Moreover, define
\[\mathcal{M}(\varPhi)=\{x\in\mathcal{B}(\hh)\;|\;\mathcal{D}(x,x)=0,
\,\mathcal{D}(x^*,x^*)=0\}.
\]

\begin{lemma}\label{lem:Schw-map}
For all $x\in\mathcal{M}(\varPhi)$ and all $a\in{\cal B}(\hh)$ we
have $\mathcal{D}(x,a)=\mathcal{D}(x^*,a)=0.$
\end{lemma}

\noindent{\bf Proof.} Since $\varPhi$ is a Schwarz map we have
$\mathcal{D}(a,a)\geq 0$ for all $a\in{\cal B}(\hh)$. Now, if
$x\in{\cal B}(\hh)$ and $\mathcal{D}(x,x)=0$, for all
$z\in\mathbb{C}$ we have $\mathcal{D}(zx+a,zx+a)\geq 0$, i.e.:
$$\overline{z}\;\mathcal{D}(x,a)+z\;\mathcal{D}(a,x)+
\mathcal{D}(a,a)\geq 0 .$$ Therefore, for all state $\omega$ in $\mathfrak{I}_1(\hh)$ we have $2\operatorname{Re}\,z\;\omega(\mathcal{D}(a,x))+
\omega(\mathcal{D}(a,a))\geq 0$ i.e., if $z=|z|\mathrm{e}^{i\theta}$,
$2|z|\operatorname{Re}\,\mathrm{e}^{i\theta}\omega(\mathcal{D}(a,x))+
\omega(\mathcal{D}(a,a))\geq 0$ for all $|z|\geq 0$ and
$\theta\in\mathbb{R}$. This implies $\mathcal{D}(a,x)=0$.
\qed

\medskip
We are now in a position to prove the non-commutative version of
the Perron-Frobenius theorem for unital Schwarz maps. Let us denote by $\mathbb{T}$ the complex torus, i.e. $\{z \in \mathbb{C}: |z|=1\}$.

\begin{theorem}\label{th:frobenius} (Perron-Frobenius) Let $\varPhi$ be an unital
irreducible Schwarz map on ${\cal B}(\hh)$ for some finite dimensional Hilbert space $\hh$. Then:
\begin{enumerate}
\item[$(1)$] $\hbox{\rm sp}(\varPhi)\cap\mathbb{T}$ is a finite
subgroup of $\mathbb{T}$,
\item[$(2)$] the eigenspace of each $\lambda\in\hbox{\rm sp}(P)\cap\mathbb{T}$
is one-dimensional and the associated eigenvector is a multiple of a unitary element of $\mathcal{A}$.
\end{enumerate}
\end{theorem}

\begin{proof} Let us first prove $(2)$. Suppose $\mathrm{e}^{i\theta}\in\hbox{\rm
sp}(\varPhi)$ for $\theta\in\mathbb{R}$ then, since ${\cal B}(\hh)$
is finite-dimensional, $\mathrm{e}^{i\theta}$ is an eigenvalue of
$\varPhi$. Let $a\in{\cal B}(\hh),\;(a\neq 0)$ such that
$\varPhi(a)=\mathrm{e}^{i\theta}a$. Clearly we have also
$\varPhi(a^*)=\mathrm{e}^{-i\theta}a^*$. Let $\omega$ be the unique
faithful $\varPhi$-invariant density. By the
Schwarz property of $\varPhi$, we have
\[
0  \leq  \omega(\varPhi(a^*a)-\varPhi(a^*)\varPhi(a)) =
\omega(a^*a-a^*\mathrm{e}^{-i\theta}\mathrm{e}^{i\theta}a)=0.
\]
It follows that $\varPhi(a^*a)=\varPhi(a^*)\varPhi(a)=a^*a$ and, by the
same argument, $\varPhi(aa^*)=\varPhi(a)\varPhi(a^*)=aa^*$, i.e.
$a\in\mathcal{M}(\varPhi)$ and $a^*a, \, aa^*$ are fixed points for $\varPhi$, therefore without loss of generality, we can assume that  $a^*a= aa^*=\unit$.

Let us consider $b\neq 0$ such that $\varPhi(b)=e^{i\theta}b$ for the same $\theta$, then both $a$ and $b$ belong to
$\mathcal{M}(\varPhi)$ and, by Lemma \ref{lem:Schw-map},
$ \varPhi(a^*b) = a^*b$,
implying $a^*b = c \unit$ for some complex constant $c$. Then $a^*(b-ca)=0$ and $b=ca$ because
$a^*$ is invertible.

\bigskip Now, let us prove item $(1)$. For any two eigenvectors $a,b \neq 0$ such that $\varPhi(a)=e^{i \theta}a$ and $\varPhi(b)=e^{i\theta^\prime}b$, one has that $ab \neq 0$ and $\varPhi(ab)=e^{i(\theta+\theta^\prime)}ab$, therefore ${\rm sp}(\varPhi) \cap \mathbb{T}$ is a subgroup of $\mathbb{T}$ and it is finite because ${\rm sp}(\varPhi)$ is finite.
\end{proof}

For the interested reader, more details on the peripheral spectrum of Schwarz maps acting on ${\rm C}^*$ and ${\rm W}^*$-algebras can be found in \cite{Bh25,Gr84}.

\section{Support of an evolved density}\label{appB}

In this section we describe the explicit formula for the support of a density $\TT_{*t}(\eta)$ evolved at time
$t$ from an initial density $\eta$ (see \cite{Ha14,HaNa}).

\begin{lemma} \label{lem:support-H-sa}
Let us assume that the operator $H$ is self-adjoint, operators $L_\ell$ are bounded and the series
$\sum_{\ell\in I}L_\ell^*L_\ell$ is strongly convergent. Then the operator $G$
\[
\operatorname{Dom}(G)=\operatorname{Dom}(H), \qquad G=-\mathrm{i}H-\frac{1}{2}\sum_{\ell\in I}L_\ell^*L_\ell
\]
generates a strongly continuous contraction semigroup $(P_t)_{t\geq 0}$ on $\hh$ and
\begin{eqnarray}
& & \kern-12truept \TT_t(x)=P_t^*xP_t \label{eq:feller-aka-dyson}\\
& & \kern-12truept + \sum_{n \geq 1}\sum_{\ell_1,\dots, \ell_n \in I}\int_{0 \leq t_1 \leq \dots \leq t_n \leq t}
\kern-16truept
P_{t-t_n}^*L_{\ell_n}^* \dots L_{\ell_1}^* P_{t_1}^*x P_{t_1} L_{\ell_1} \dots L_{\ell_n}P_{t-t_n} dt_1 \dots dt_n\, ,
\nonumber
\end{eqnarray}
where the series converges in norm. \\
For all $\psi\in\hh$ and $t>0$ the support of the density $\TT_{*t}(|\psi\rangle\langle\psi|)$ is
the closure of
\begin{equation}\label{eq:supp-unbd}
\operatorname{span}_{\mathbb{C}}\left\{\,
P_t\psi,
\,P_{t_1} L_{\ell_1} \dots L_{\ell_n}P_{t-t_n}\psi\,\mid\, n\geq 1, \, 0 \leq t_1 \leq \dots \leq t_n \leq t\right\}\, .
\end{equation}
\end{lemma}

\begin{proof}
The operator $G$ generates a strongly continuous semigroup by the Bounded Perturbation Theorem (see \cite{EnNa} Ch. III Th. 1.3).
It is a contraction semigroup since, for all $u\in \operatorname{Dom}(G)$, we have
\[
\frac{\mathrm{d}}{\mathrm{d}t}\Vert P_t u \Vert^2 = \left\langle G P_t u, P_t u\right\rangle
+ \left\langle  P_t u, G P_t u\right\rangle = -\sum_{\ell\in I}\left\Vert L_\ell P_t u\right\Vert^2 \leq 0
\]
and so $\Vert P_t u \Vert^2\leq \Vert u \Vert^2$ for all $t\geq 0$. The inequality extends to all $u\in\hh$
by density of $\operatorname{Dom}(G)$.

For all $u,v\in \operatorname{Dom}(G)$ also $P_tv,P_tu$ belong to $\operatorname{Dom}(G)$ for all $t\geq 0$ and, for $0<s<t$
\[
\frac{\mathrm{d}}{\mathrm{d}s}\left\langle P_{t-s}v,\TT_s(x)P_{t-s}u\right\rangle
= \sum_{\ell\in I} \left\langle L_\ell P_{t-s}v,\TT_s(x) L_\ell  P_{t-s}u\right\rangle.
\]
Integration for $s\in [0,t]$ yields
\[
\left\langle v,\TT_t(x)u\right\rangle
= \left\langle P_t v,x P_t u\right\rangle
+ \sum_{\ell\in I}\int_0^t \left\langle L_\ell P_{t-s}v,\TT_s(x) L_\ell  P_{t-s}u\right\rangle\mathrm{d}s.
\]
Iterating this formula for $\TT_s(x)$ we find
\begin{eqnarray*}
\left\langle v,\TT_t(x)u\right\rangle
& = & \left\langle P_t v,x P_t u\right\rangle
+ \sum_{\ell\in I}\int_0^t \left\langle L_\ell P_s L_\ell P_{t-s}v,x P_s L_\ell  P_{t-s}u\right\rangle\mathrm{d}s \\
& + & \sum_{\ell,\ell'\in I}\int_0^t \mathrm{d}s \int_0^s \mathrm{d}r
\left\langle L_{\ell'}P_{s-r}L_\ell P_{t-s}v,\TT_r(x) L_{\ell'}P_{s-r}L_\ell P_{t-s}u\right\rangle\, .
\end{eqnarray*}
Iterating $n$ times and taking the limit as $n$ tends to infinity we find \eqref{eq:feller-aka-dyson}. \\
Finally, \eqref{eq:supp-unbd} follows noting that a vector $w$ is orthogonal to the support
of $\TT_{*t}(|\psi\rangle\langle\psi|)$ if and only if $\langle \psi, \TT_t(|w\rangle\langle w|)\psi\rangle=0$,
namely, by \eqref{eq:feller-aka-dyson},
\begin{equation}\label{eq:iterat-PL}
|\langle w, P_t \psi \rangle|^2=0, \quad |\langle w, P_{t_1} L_{\ell_1} \dots L_{\ell_n}P_{t-t_n} \psi \rangle|^2=0
\end{equation}
(where we used the positivity of each addend in the series and the norm continuity of the function $(t_1,\dots, t_n) \mapsto P_{t-t_n}L_{\ell_n}\cdots P_{t_2-t_1}L_{\ell_1}P_{t_1}\psi$). This completes the proof.
\end{proof}

The support of an evolved state at time $t$ for a uniformly continuous QMS is characterized more explicitly by the following Lemma. Let us recall the notation $\delta_a(b)=[a,b]$ for every $a,b \in {\cal B}(\hh)$.

\begin{lemma} \label{lem:supp}
Let us assume that $\TT$ is uniformly continuous and let us consider a normalized vector $\psi \in \hh$; the support of $\TT_{*t}(\ket{\psi}\bra{\psi})$ for $t >0$ is given by $P_t \mathcal{S}(\psi)$, where
$$
\mathcal{S}(\psi)={\rm span}_{\mathbb{C}}\{\psi, \delta^{m_1}_G(L_{\ell_1}) \cdots
\delta^{m_n}_G(L_{\ell_n}) \psi\mid  n \in \mathbb{N}^*, \, m_1,.., m_n \in \mathbb{N}, \ell_1, .., \ell_n \in I\}.
$$
\end{lemma}
\begin{proof}
Fix $t >0$; using the expression of the semigroup in Lemma \ref{lem:support-H-sa} and formula \eqref{eq:iterat-PL},
one sees that a vector $w \in \hh$ is in the kernel of $\TT_{*t}(\ket{\psi}\bra{\psi})$
if and only if
\[
|\langle w, P_t \psi \rangle|^2=0, \quad |\langle w,  P_{t_1} L_{\ell_1} \dots L_{\ell_n}P_{t-t_n}  \psi \rangle|^2=0
\]
for every $0 \leq t_1 \leq t_2 \leq \dots \leq t_n \leq t$.

Notice that, in fact, the function $(t_1,\dots, t_n) \mapsto P_{t_1} L_{\ell_1} \dots L_{\ell_n}P_{t-t_n}$ is analytic,
therefore it is identically equal to zero if and only if all its derivative at $t_1=\dots =t_n=t$ are equal to $0$ and
we are $w$ belongs to the orthogonal of $P_t\mathcal{S}(\psi)$. Conversely, if $w$ belongs to the orthogonal of $P_t\mathcal{S}(\psi)$
then all derivatives of the analytic function
$(t_1,\dots, t_n) \mapsto \langle w,  P_{t_1} L_{\ell_1} \dots L_{\ell_n}P_{t-t_n}  \psi \rangle$
at $t_1=\dots =t_n=t$ are equal to $0$ are zero and so it is identically equal to zero and $w$ is orthogonal to $P_t\mathcal{S}(\psi)$.
\end{proof}

%

%
%

\vfill\eject

\end{document}